\documentclass[sigconf]{acmart}

\copyrightyear{2017} 
\acmYear{2017} 
\setcopyright{acmlicensed}
\acmConference{WI '17}{August    23-26, 2017}{Leipzig, Germany}\acmPrice{15.00}\acmDOI{10.1145/3106426.3106484}
\acmISBN{978-1-4503-4951-2/17/08}

\usepackage{hyperref}

\usepackage{comment}

\usepackage{array}

\usepackage{tabularx}
\usepackage{xcolor}
\usepackage{colortbl}
\usepackage{arydshln}
\usepackage{multirow}

\definecolor{LightGreen}{rgb}{0.9, 1, 0.9}
\definecolor{LightRed}{rgb}{1, 0.9, 0.9}
\definecolor{LightBlue}{rgb}{0.9, 0.9, 1}
\definecolor{DarkGreen}{rgb}{0, 0.9, 0}
\definecolor{DarkRed}{rgb}{0.9, 0, 0}
\definecolor{DarkYellow}{rgb}{1, 0.9, 0}
\definecolor{Highlight}{rgb}{0.618, 0, 0}
\definecolor{NoteHighlight}{rgb}{1, 1, 0.5}
\definecolor{DoneHighlight}{rgb}{1, 0.8, 0.8}
\definecolor{DarkGray}{gray}{0.5}
\definecolor{LightGray}{gray}{0.9}
\definecolor{Gray}{gray}{0.75}

\usepackage{caption}
\usepackage[subrefformat=parens, labelformat=parens, caption=false]{subfig}

\usepackage[noend, algo2e, linesnumbered, ruled]{algorithm2e}

\SetCommentSty{mycommfont}

\usepackage{xfrac}
\usepackage{mathtools}

\newcommand{\pone}{\textsf{MaxGroups}\xspace}
\newcommand{\ptwo}{\textsf{MinOutJoin}\xspace}

\providecommand{\abs}[1]{\left\lvert#1\right\rvert}

\providecommand{\expr}[1]{\left(#1\right)}

\providecommand{\tuple}[1]{\left(#1\right)}

\providecommand{\set}[1]{\left\{#1\right\}}
\providecommand{\condset}[2]{\left\{#1\,\middle|\,#2\right\}}

\providecommand{\ceil}[1]{\left\lceil#1\right\rceil}

\providecommand{\highlight}[1]{\textcolor{Highlight}{#1}}
\providecommand{\NULL}{\setlength{\fboxsep}{1pt}\colorbox{LightGray}{\texttt{NULL}}}

\newcommand{\topk}{{top}}
\newcommand{\cover}{{cover}}

\newcommand{\join}{{join}}

\begin{document}
\begin{sloppy}

\title{Preference-driven Similarity Join}

\thanks{
This work is supported in part by the NSERC Discovery Grant program, the Canada Research Chair program, the NSERC Strategic Grant program, and the SFU President's Research Start-up Grant. All opinions, findings, conclusions and recommendations in this paper are those of the authors and do not necessarily reflect the views of the funding agencies.
\\
}

\author{Chuancong Gao}
\affiliation{Simon Fraser University}
\email{cga31@sfu.ca}
\author{Jiannan Wang}
\affiliation{Simon Fraser University}
\email{jnwang@sfu.ca}
\author{Jian Pei}
\affiliation{Simon Fraser University}
\affiliation{Huawei Technologies Co., Ltd.}
\email{jpei@cs.sfu.ca}
\author{Rui Li}
\affiliation{Yahoo! Labs}
\email{ruililab@yahoo-inc.com}
\author{Yi Chang}
\affiliation{Huawei Research America}
\email{yichang@acm.org}

\begin{abstract}
Similarity join, which can find similar objects (e.g., products, names, addresses) across different sources, is powerful in dealing with variety in big data, especially web data. Threshold-driven similarity join, which has been extensively studied in the past, assumes that a user is able to specify a similarity threshold, and then focuses on how to efficiently return the object pairs whose similarities pass the threshold. We argue that the assumption about a well set similarity threshold may not be valid for two reasons. The optimal thresholds for different similarity join tasks may vary a lot. Moreover, the end-to-end time spent on similarity join is likely to be dominated by a back-and-forth threshold-tuning process.

In response, we propose preference-driven similarity join. The key idea is to provide several \emph{result set preferences}, rather than a range of thresholds, for a user to choose from. Intuitively, a result set preference can be considered as an objective function to capture a user's preference on a similarity join result.
Once a preference is chosen, we automatically compute the similarity join result optimizing the preference objective.
As the proof of concept, we devise two useful preferences and propose a novel preference-driven similarity join framework coupled with effective optimization techniques. Our approaches are evaluated on four real-world web datasets from a diverse range of application scenarios. The experiments show that preference-driven similarity join can achieve high-quality results without a tedious threshold-tuning process.
\end{abstract}

\maketitle

\section{Introduction}
\label{sec:intro}

A key characteristic of big data is \emph{variety}. Data (especially web data) often comes from different sources and the value of data can only be extracted by integrating various sources together. Similarity join, which finds similar objects (e.g., products, people, locations) across different sources, is a powerful tool for tackling the challenge.

For example, suppose a data scientist collects a set of restaurants from Groupon.com and would like to know which restaurants are highly rated on Yelp.com. Since a restaurant may have different representations in the two data sources (e.g., \textsf{``10 East Main Street"} vs. \textsf{``10 E Main St., \#2"}), she can use similarity join to find these similar restaurant pairs and integrate the two data sources together.

\emph{Threshold-driven similarity join} has been extensively studied in the past~\cite{DBLP:conf/vldb/ArasuGK06,DBLP:conf/icde/ChaudhuriGK06,DBLP:conf/www/BayardoMS07,DBLP:conf/www/XiaoWLY08,DBLP:journals/pvldb/XiaoWL08,DBLP:conf/icde/XiaoWLS09,DBLP:journals/pvldb/WangLF10,DBLP:journals/pvldb/LiDWF11,DBLP:conf/sigmod/WangLF12,DBLP:journals/pvldb/JiangLFL14,DBLP:conf/sigmod/VernicaCL10,DBLP:journals/pvldb/MetwallyF12,DBLP:conf/icde/DengLHWF14,DBLP:conf/sigmod/LuLWLW13,DBLP:journals/pvldb/BourosGM12}.
To use it, one has to go through three steps: (a) selecting a similarity function (e.g., Jaccard), (b) selecting a threshold (e.g., $0.8$), and (c) running a similarity join algorithm to find all object pairs whose similarities are at least $0.8$. The existing studies are mainly focused on Step (c). However, both Steps (a) and (b) deeply implicate humans in the loop, which can be orders of magnitude slower than conducting the actual similarity join.

One may argue that, in reality, humans are able to {quickly} select an appropriate similarity function and a corresponding threshold for a given similarity join task. For choosing similarity function, this may be true because humans can understand the semantics of each similarity function and choose the one that meets their needs.

However, selecting an appropriate threshold may be far from easy. It is extremely difficult for humans to figure out the effect of different thresholds on result quality. Choosing a good threshold depends on not only the specified similarity function but also the underlying data.
We conduct an empirical analysis on the optimal thresholds for a diverse range of similarity join tasks, where the optimal thresholds maximize $F_1$-scores~{\cite{DBLP:journals/pvldb/WangLYF11}}. Table~\ref{tab:exp/optimal} shows the results (details of the experiment are in Section~\ref{sec:exp}). We find that the optimal thresholds for the tasks are quite different. Even for the same similarity function, the optimal thresholds may still vary a lot. For example, the optimal threshold of a record-linkage task on the \textsf{Restaurants} dataset is $0.6$, which differs a lot from the optimal threshold of $0.34$ on the \textsf{Scholar-DBLP} dataset using the same similarity function.

\begin{table}[t]
\small
\centering

\caption{Example of optimal thresholds w.r.t.\ various tasks.}
\label{tab:exp/optimal}

\begin{tabular}{l|l|c:c}
Dataset & Task & Optimal Threshold & Similarity \\
\hline
\textsf{Wiki Editors} & Spell checking & $0.625$ & Jaccard \\
\textsf{Restaurants} & Record linkage & $0.6$ & Jaccard \\
\textsf{Scholar-DBLP} & Record linkage & $0.34$ & Jaccard \\
\textsf{Wiki Links} & Entity matching & $0.9574$ & Tversky \\
\end{tabular}
\end{table}

To solve this problem, one idea may be to label some data and then use the labeled data as the {ground truth of matching pairs} to tune the threshold. However, human labeling is error-prone and time-consuming, which significantly increases the (end-to-end) time of data integration or cleaning using similarity join.

In this paper, we tackle this problem from a different angle -- \emph{can we achieve high-quality similarity join results without requiring humans to label any data or specifying a similarity threshold?} 
{
Our key insight is inspired by the concept of preference in areas like economics, which is an ordering of different alternatives (results)~{\cite{Scholar:preference}}. Taking Yelp.com as an example, the restaurants can be presented in different ways such as by distance, price, or rating. The different ordering may meet different search intents. A user needs to evaluate her query intent and choose the most suitable ranking accordingly.
}
{Similarly,} when performing a similarity join, we seek to provide a number of \emph{result set preferences} for a user to select from. Intuitively, a result set preference can be thought of as an objective function to capture how much a user likes a similarity join result. Once a particular preference is chosen, we automatically tune the threshold such that the preference objective function is maximized, and then we return the corresponding similarity join result.
We call this new similarity join model \emph{preference-driven similarity join}. Compared to the traditional threshold-driven similarity join, this new model does not need any labeled data. 

{
As a proof of concept, our paper proposes two preferences from different yet complementary perspectives.
The first preference \mbox{\pone} groups the joined pairs where each group is considered as an entity across two data sources, and returns the join result having the largest number of groups. The second preference \mbox{\ptwo} balances between matching highly similar pairs and joining many possibly matching pairs, and favors the join result minimizing the outer-join size.
}
According to our experiments on various datasets with ground-truth, the preference-driven approach can achieve optimal or nearly optimal $F_1$-scores on different tasks without knowing anything about the optimal thresholds.

Given a result set preference, a challenge is how to develop an efficient algorithm for preference-driven similarity join. This problem is more challenging than the traditional threshold-driven similarity join because it involves one additional step: finding the best threshold such that a preference is maximized. 
The brute-force method needs to compute the similarity values for all the pairs. It is highly inefficient even for datasets of moderate size. We solve this problem by developing a novel similarity join framework along with effective optimization techniques. The experimental results show that the proposed framework achieves several orders of magnitude speedup over the brute-force method.

The rest of the paper is organized as follows.
In Section~\ref{sec:problem}, we formally define the problem of preference-driven similarity join. In Section~\ref{sec:preference}, we design two result set preferences from different perspectives. In Section~\ref{sec:framework}, we propose a preference-driven similarity join framework, and develop efficient algorithm for set-based similarity functions. In Section~\ref{sec:exp}, we evaluate our approach on four real-world web datasets from a diverse range of applications. The results suggest that preference-driven similarity join is a promising idea to tackle the threshold selection problem for similarity join, and verify that our method is effective and efficient. We review related work in Section~\ref{sec:related} and conclude the paper in Section~\ref{sec:conclusion}.

\section{Problem Definition}
\label{sec:problem}

Let $R$ and $S$ be two sets of objects, $\mathbb{C^+}$ denote the ground-truth that is the set of pairs in $R \times S$ that should be joined/matched{, and $\mathbb{C^-}$ denote the remaining pairs, i.e., $\mathbb{C^-} = (R \times S)\setminus \mathbb{C^+}$.
We call the pairs in $\mathbb{C^+}$ and $\mathbb{C^-}$ matching pairs and non-matching pairs, respectively.
{In general, $\mathbb{C^+}$ and $\mathbb{C^-}$ are assumed unknown.}
}
Figure~\ref{fig:toy/table} shows a toy running example of $R$ and $S$, as well as the ground-truth $\mathbb{C^+}$.
\begin{figure}[t]
\centering
\subfloat{
\begin{minipage}{0.4\linewidth}
\includegraphics[width=\linewidth]{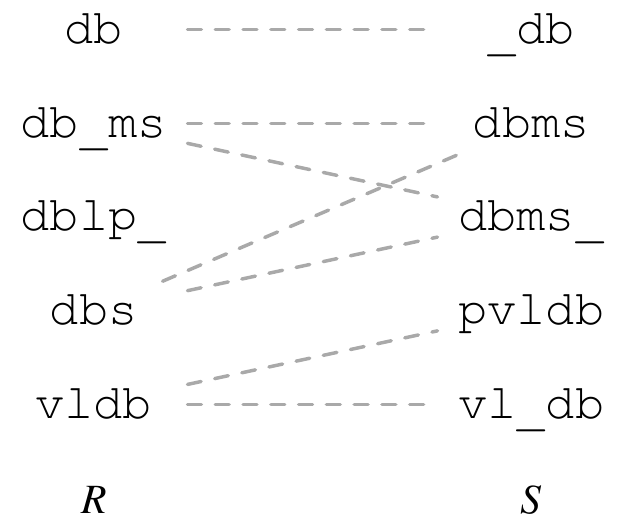}
\end{minipage}
}
\quad$\Longrightarrow$\quad
\small
\captionsetup[subfigure]{labelformat=empty}
\setcounter{subfigure}{0}
\subfloat[$(r, s)$]{
\begin{tabular}{c}
$\tuple{\texttt{\textcolor{Highlight}{db\_ms}}, \texttt{\textcolor{Highlight}{dbms\_}}}$ \\
\hline
$\tuple{\texttt{\textcolor{Highlight}{db}\_\highlight{ms}}, \texttt{\textcolor{Highlight}{dbms}\ }}$ \\
$\tuple{\texttt{\textcolor{Highlight}{vldb}\ }, \texttt{p\textcolor{Highlight}{vldb}}}$ \\
$\tuple{\texttt{\textcolor{Highlight}{vldb}\ }, \texttt{\textcolor{Highlight}{vl}\_\highlight{db}}}$ \\
\hline
$\tuple{\texttt{\textcolor{Highlight}{dbs}\ \ }, \texttt{\textcolor{Highlight}{db}m\highlight{s}\ }}$ \\
\hline
$\tuple{\texttt{\textcolor{Highlight}{db}\ \ \ }, \texttt{\_\textcolor{Highlight}{db}\ \ }}$ \\
$\tuple{\texttt{\textcolor{Highlight}{dblp}\_}, \texttt{\textcolor{Highlight}{p}v\textcolor{Highlight}{ldb}}}$ \\
$\tuple{\texttt{\textcolor{Highlight}{dbl}p\highlight{\_}}, \texttt{v\textcolor{Highlight}{l\_db}}}$ \\
\hline
$\tuple{\texttt{\textcolor{Highlight}{db\_}ms}, \texttt{\textcolor{Highlight}{\_db}\ \ }}$ \\
$\tuple{\texttt{\textcolor{Highlight}{db}lp\highlight{\_}}, \texttt{\textcolor{Highlight}{\_db}\ \ }}$ \\
$\tuple{\texttt{\textcolor{Highlight}{dbs}\ \ }, \texttt{\textcolor{Highlight}{db}m\highlight{s}\_}}$ \\
\hline
$\tuple{\texttt{\textcolor{Highlight}{db}\ \ \ }, \texttt{\textcolor{Highlight}{db}ms\ }}$ \\
$\tuple{\texttt{\textcolor{Highlight}{db}s\ \ }, \texttt{\_\textcolor{Highlight}{db}\ \ }}$ \\
\end{tabular}
}
~
\subfloat[$sim(r, s)$]{
\begin{tabular}{c}
\multirow{1}{*}{$1$} \\
\hline
\multirow{3}{*}{${0.8}$} \\
\\
\\
\hline
${0.75}$ \\
\hline
\multirow{3}{*}{$0.667$} \\
 \\
 \\
\hline
\multirow{3}{*}{$0.6$} \\
 \\
 \\
\hline
\multirow{2}{*}{$0.5$} \\
 \\
\end{tabular}
}
\caption{Example of $join(\theta)$, where $\theta = 0.5$. The optimal threshold is $0.75$. The edges represent the ground-truth $\mathbb{C^+}$.}\label{fig:toy/table}
\end{figure}

Let $sim: (r, s) \in R \times S \rightarrow [0, 1]$ denote a similarity function.
\begin{definition}[Threshold-driven similarity join]
Given a similarity function $sim$ and a threshold $\theta$, return all the object pairs $(r, s)$ whose similarity values are at least $\theta$, that is,
\[join(R, S, sim, \theta) = \condset{(r, s) \in R \times S}{sim(r, s) \ge \theta}
\rlap{\qed}
\]
\end{definition}
If $R$, $S$, and $sim$ are clear from the context, we write $join(\theta)$ for the sake of brevity.

We can regard $join(\theta)$ as a classifier, where the pairs returned by the function are the ones classified as positive, and the rest pairs $\expr{R \times S} \setminus \expr{join(\theta)}$ classified as negative, that is, not-matching.
Figure~\ref{fig:toy/table} shows an example of $join(0.5)$ using Jaccard similarity. Here, for simplicity we tokenize a string into a set of characters. For example, $jaccard(\texttt{\highlight{db}lp\highlight{\_}}, \texttt{\highlight{\_db}}) = \frac{\abs{r\cap s}}{\abs{r\cup s}} = \frac{\abs{\{\texttt{\highlight{\_}}, \texttt{\highlight{d}}, \texttt{\highlight{b}}\}}}{\abs{\{\texttt{\highlight{\_}}, \texttt{\highlight{d}}, \texttt{\highlight{b}}, \texttt{l}, \texttt{p}\}}} = \frac{3}{5} = 0.6$. 

Similarity join seeks to find a threshold that leads to the best result quality. Theoretically, there are an infinite number of thresholds to choose from. {However, we only need to consider a finite subset of the possible thresholds, which is the set of the similarity values of all object pairs in $R\times S$, i.e., $\condset{sim(r,s)}{r \in R, s \in S}$}, because, for any threshold not in the finite set, there is always a smaller threshold in the finite set having the same join result. {For example, threshold $0.9$ is not in the finite set in the example in Figure~\ref{fig:toy/table} but $join(0.9) = join(0.8)$.}

Threshold tuning is time consuming and labor intensive. Therefore, we develop preference-driven similarity join to overcome the limitations. Before a formal definition, we first introduce the concept of \emph{result set preference}, to capture the user preferences on a similarity join result. Formally, a result set preference is a score function~$h$: $(R, S, sim, \theta) \rightarrow \mathbb{R}$, where $\mathbb{R}$ is the set of real numbers. {Obviously, a result set is determined by $R$, $S$, $sim$ and $\theta$.  The result set preference gives a score on how well the result set meets a user's preference.  The higher the score, the better.} If $R$, $S$, and $sim$ are clear from the context, we write $h(\theta)$ for the sake of brevity.
\begin{definition}[Preference-driven Similarity Join]
Given a similarity function $sim$ and a result set preference $h$, return the most preferred result
$join(\theta^*)$ where $\theta^*$ is the largest threshold {in the finite set} maximizing $h$.
\qed
\end{definition}

For the ease of presentation, we introduce some notations.  First, we denote by
$\join^=(\theta) = \condset{(r, s) \in join(\theta)}{sim(r, s) = \theta}$ the subset of joined pairs w.r.t.\ similarity $\theta$.  For example, in Figure~\ref{fig:toy/table}, $\join^=(0.75) = \set{\tuple{\texttt{\textcolor{Highlight}{dbs}}, \texttt{\textcolor{Highlight}{db}m\highlight{s}}}}$.
Second, we denote by
$\cover^R(\theta) = \condset{r}{\exists s: (r, s) \in join(\theta)}$
the set of objects in $R$ that are joined when the similarity threshold is $\theta$.
Similarly, we have $\cover^S(\theta) = \condset{s}{\exists r: (r, s) \in join(\theta)}$.
For example, in Figure~\ref{fig:toy/table}, $\cover^R(0.75) = \set{\texttt{{db\_ms}}, \texttt{{dbs}}, \texttt{{vldb}}}$. Last, for $r \in R$, we denote by $top^S(r)$ the set of most similar object(s) in $S$ including ties. Similarly, we have $top^R(s)$.
For example, in Figure~\ref{fig:toy/table},
$top^S(\texttt{{vldb}}) = \set{\texttt{p{vldb}}, \texttt{{vl}\_{db}}}$.

\section{Two Result Set Preferences}
\label{sec:preference}

As a proof of concept, we present two result set preferences. 

\subsection{\pone: Maximum Number of Non-Trivial Connected Components}
\label{sec:preference/cluster}
Our first preference, \pone, partitions objects into different groups without any prior knowledge.
For a similarity threshold $\theta$, we construct a bipartite graph $G_{\theta} = \tuple{U^R, V^S, join(\theta)}$, where $U^R$ and $V^S$ are two disjoint sets of nodes representing the objects in $R$ and $S$, respectively, and every pair in $join(\theta)$ defines an edge.
As indicated in~\cite{DBLP:conf/dmkd/MongeE97}, each connected component in the bipartite graph is an entity, where the objects in the same connected component are the entity's different representations.

\pone prefers the similarity join result with more non-trivial connected components (i.e., connected components with at least two nodes, one in $R$ and another in $S$). The intuition is that heuristically we want to match as many entities as possible across $R$ and $S$. Let $J(G_{\theta})$ denote the set of non-trivial connected components in a bipartite $G_{\theta}$, we define result set preference \pone as
\begin{align*}
h_c(\theta) = \abs{J(G_{\theta})}
\end{align*}
Figure~\ref{fig:toy/outerJoin} gives an example of $h_c$ on our toy dataset in Figure~\ref{fig:toy/table}.

\begin{figure}[t]
\small
\centering
\captionsetup[subfigure]{labelformat=empty}
\subfloat[$(r, s) \in R \times S$]{
\begin{tabular}{c}
$\tuple{\texttt{\textcolor{Highlight}{db\_ms}}, \texttt{\textcolor{Highlight}{dbms\_}}}$ \\
\hline
$\tuple{\texttt{\textcolor{Highlight}{db}\_\highlight{ms}}, \texttt{\textcolor{Highlight}{dbms}\ }}$ \\
$\tuple{\texttt{\textcolor{Highlight}{vldb}\ }, \texttt{p\textcolor{Highlight}{vldb}}}$ \\
$\tuple{\texttt{\textcolor{Highlight}{vldb}\ }, \texttt{\textcolor{Highlight}{vl}\_\highlight{db}}}$ \\
\hline
$\tuple{\texttt{\textcolor{Highlight}{dbs}\ \ }, \texttt{\textcolor{Highlight}{db}m\highlight{s}\ }}$ \\
\hline
$\tuple{\texttt{\textcolor{Highlight}{db}\ \ \ }, \texttt{\_\textcolor{Highlight}{db}\ \ }}$ \\
$\tuple{\texttt{\textcolor{Highlight}{dblp}\_}, \texttt{\textcolor{Highlight}{p}v\textcolor{Highlight}{ldb}}}$ \\
$\tuple{\texttt{\textcolor{Highlight}{dbl}p\highlight{\_}}, \texttt{v\textcolor{Highlight}{l\_db}}}$ \\
\hline
$\tuple{\texttt{\textcolor{Highlight}{db\_}ms}, \texttt{\textcolor{Highlight}{\_db}\ \ }}$ \\
$\tuple{\texttt{\textcolor{Highlight}{db}lp\highlight{\_}}, \texttt{\textcolor{Highlight}{\_db}\ \ }}$ \\
$\tuple{\texttt{\textcolor{Highlight}{dbs}\ \ }, \texttt{\textcolor{Highlight}{db}m\highlight{s}\_}}$ \\
\hline
$\tuple{\texttt{\textcolor{Highlight}{db}\ \ \ }, \texttt{\textcolor{Highlight}{db}ms\ }}$ \\
$\tuple{\texttt{\textcolor{Highlight}{db}s\ \ }, \texttt{\_\textcolor{Highlight}{db}\ \ }}$ \\
\hline
$\cdots$ \\
\end{tabular}
}
~
\subfloat[$\theta$]{
\begin{tabular}{c}
\multirow{1}{*}{$1$} \\
\hline
\multirow{3}{*}{$0.8$} \\
\\
\\
\hline
$0.75$ \\
\hline
\multirow{3}{*}{${0.667}$} \\
 \\
 \\
\hline
\multirow{3}{*}{$0.6$} \\
 \\
 \\
\hline
\multirow{2}{*}{$0.5$} \\
\\
\hline
$\cdots$ \\
\end{tabular}
}
~
\subfloat[$h_c(\theta)$]{
\begin{tabular}{c}
$1$ \\
\hline
\multirow{3}{*}{$2$} \\
 \\
 \\
\hline
$2$ \\
\hline
\multirow{3}{*}{$\mathbf{3}$} \\
 \\
 \\
\hline
\multirow{3}{*}{$1$} \\
 \\
 \\
\hline
\multirow{2}{*}{$1$} \\
 \\
\hline
$\cdots$ \\
\end{tabular}
}
~
\subfloat[$h_o(\theta)$]{
\begin{tabular}{c}
$1$ \\
\hline
\multirow{3}{*}{$\mathbf{2}$} \\
 \\
 \\
\hline
$\mathbf{2}$ \\
\hline
\multirow{3}{*}{$\mathbf{2}$} \\
 \\
 \\
\hline
\multirow{3}{*}{$-1$} \\
 \\
 \\
\hline
\multirow{2}{*}{$-3$} \\
 \\
\hline
$\cdots$ \\
\end{tabular}
}

\caption{Example of $h_c$ and $h_o$.}\label{fig:toy/outerJoin}\label{fig:toy/oneToMany}
\end{figure}

Given a similarity join result $join(\theta)$, the time complexity of computing $h_c(\theta)$ is $O(\abs{R} + \abs{S} + \abs{join(\theta)})$ by simply computing the connected components of the bipartite graph $G_{\theta}$.
Since we need to compute the preference scores for multiple similarity thresholds $\theta$, there are opportunities to reduce the computational cost. We will discuss the details in Section~{\ref{sec:algorithm/incscore}}.

\subsection{\ptwo: Minimum Outer-Join Size}
\label{sec:preference/outerJoin}

We make the following observation.  On the one hand, if we set a too high similarity threshold and thus be too strict in similarity, many objects may not be matched with their counterparts due to noise. The extreme case is that, if we set the similarity threshold to $1$, only those perfectly matching objects are joined.  On the other hand, if we set a too low similarity threshold and thus be too loose in similarity, many not-matching objects may be joined by mistake. The extreme case is that, by setting the similarity threshold to $0$, every pair of objects in the two sets are joined. 

We need to find a good balance between the two and strive to a good tradeoff.  Technically, full outer-join includes both joined entries and those not joined (by matching with $\NULL$).  The size of the full outer-join is jointly determined by the number of objects joined and the number of objects not joined. The two numbers trade off each other.  Therefore, if we minimize the size of the full outer-join, we reach a tradeoff between the two ends.  This is the intuition behind our second preference, \ptwo.  

\begin{figure}[t]
\small
\centering
\captionsetup[subfigure]{labelformat=empty}
\begin{tabular}{c}
$\tuple{\texttt{\textcolor{Highlight}{db}\_\highlight{ms}~}, \texttt{\textcolor{Highlight}{dbms}~}}$ \\
$\tuple{\texttt{\textcolor{Highlight}{db\_ms}~}, \texttt{\textcolor{Highlight}{dbms\_}}}$ \\
$\tuple{\texttt{\textcolor{Highlight}{dbs}\ \ \ }, \texttt{\textcolor{Highlight}{db}m\highlight{s}~}}$ \\
$\tuple{\texttt{\textcolor{Highlight}{vldb}\ \ }, \texttt{p\textcolor{Highlight}{vldb}}}$ \\
$\tuple{\texttt{\textcolor{Highlight}{vldb}\ \ }, \texttt{\textcolor{Highlight}{vl}\_\highlight{db}}}$
\end{tabular}
$\cup$
\begin{tabular}{c}
$\tuple{\texttt{db\ \ \ }, \NULL}$ \\
$\tuple{\texttt{dblp\_}, \NULL}$
\end{tabular}
$\cup$
\begin{tabular}{c}
$\tuple{\NULL, \texttt{\_db}}$
\end{tabular}

\caption{Example of $outjoin(\theta)$, where $\theta = 0.75$.}\label{fig:toy/outerJoin/output}
\end{figure}

The full outer similarity join result w.r.t.\ a similarity threshold is 
\begin{align*}
outjoin(\theta) = \join(\theta) &\cup \condset{(r, \NULL)}{r \in R \setminus \cover^R(\theta)} \\
&\cup \condset{(\NULL, s)}{s \in S \setminus \cover^S(\theta)}
\end{align*}
where $\condset{(r, \NULL)}{r \in R \setminus \cover^R(\theta)}$ is the set of objects in $R$ that are not joined, and $\condset{(\NULL, s)}{s \in S \setminus \cover^S(\theta)}$ is the set of objects in $S$ that are not joined.
Figure~\ref{fig:toy/outerJoin/output} illustrates an example of a full outer-join.

We define our preference \ptwo as
\begin{align*}
h_o(\theta) = \abs{R}+\abs{S}- \abs{outjoin(\theta)} = \abs{\cover^R(\theta)} + \abs{\cover^S(\theta)} - \abs{\join(\theta)}
\end{align*}
where $\abs{R}+\abs{S}$ is a constant given $R$ and $S$. Figure~\ref{fig:toy/outerJoin} gives an example of $h_o$ on our toy dataset.

This preference gives a penalty when multiple objects in a set are joined with multiple objects in the other set. Joining $x$ objects in $R$ and $y$ objects in $S$ results in $x \cdot y$ pairs in the full outer-join. Not joining them results in at most $x + y$ of pairs in the full outer-join. When $x > 1$ and $y > 1$, we have $x \cdot y \ge x + y$.

Given a threshold $\theta$, it is straightforward to compute $cover^R(\theta)$, $cover^S(\theta)$ and $join(\theta)$ according to their definitions by scanning the join result $join(\theta)$. The time complexity of computing $h_o(\theta)$ is $O(\abs{join(\theta)})$ for each $\theta$.
Since we need to compute the preference scores for multiple similarity thresholds $\theta$, there are opportunities to reduce the computational cost. We will discuss the details in Section~{\ref{sec:algorithm/incscore}}.

\section{Algorithm Framework}\label{sec:framework}

In this section, we present an efficient framework for the preference-driven similarity join problem. A brute-force solution is to compute the similarities for all the object pairs, calculate the preference score w.r.t.\ each possible threshold, and return the similarity join result with the highest preference score. This brute-force method may be inefficient.  Computing similarities for all pairs is often prohibitive. The number of possible thresholds can be very large, $\abs{R}\times \abs{S}$ in the worst case. It is crucial to reduce the cost involved in this process.

To tackle the challenges, we propose a preference-driven similarity join framework in Algorithm~\ref{algo:general-1}. Central to the framework are four key functions. Function \textsc{PivotalThresholds} (Section~\ref{threshold}) identifies a small set of thresholds $\Theta$, called the pivotal thresholds, that are guaranteed to cover the best preference score obtained from all the possible thresholds. Function \textsc{IncrementalSimJoin} (Section~\ref{sec:algorithm/incsimjoin}) checks the pivotal thresholds in value descending order and incrementally computes the similarity join result for each threshold. We propose a new optimization technique, called lazy evaluation, to further improve the efficiency.  Function \textsc{IncrementalScore} (Section~\ref{sec:algorithm/incscore}) computes the preference score for each threshold. It is possible to reduce the cost by computing the scores incrementally when checking the pivotal thresholds in value descending order. Function \textsc{EarlyTermination} (Section~\ref{termination}) determines whether we can stop checking the remaining pivotal thresholds by comparing the upper bound $\widehat{h}(\theta_i)$ with currently best score $h(\theta^*)$ once a similarity join result $\join(\theta_i)$ is computed. 

\begin{algorithm2e}[t]
\small
\DontPrintSemicolon

\SetKwInput{Prototype}{Function}
\SetKwInput{Input}{Input}
\SetKwInput{Output}{Output}

\SetKw{LogicalAnd}{and}
\SetKw{LogicalOr}{or}
\SetKw{LogicalNot}{not}
\SetKw{Continue}{continue}
\SetKw{Break}{break}
\SetKw{True}{true}
\SetKw{False}{false}
\Input{objects $R$ and $S$, similarity function $sim$, preference $h$}
\Output{the most preferred join result $\join(\theta^*)$}
\BlankLine
	$\Theta \leftarrow \textsc{\highlight{PivotalThresholds}}(R, S, sim)$\;
    \ForEach{threshold $\theta_i \in \Theta$ in descending order}{
	    $\join^=(\theta_i) \leftarrow \textsc{\highlight{IncrementalSimJoin}}(\theta_{i-1}, \theta_i)$\;
	    $\join(\theta_i) \leftarrow \join(\theta_{i-1}) \cup \join^=(\theta_i)$\;
        \lIf{$h(\theta_i) > h(\theta^*)$}{
            $\theta^* \leftarrow \theta_i$
            \tcp*[f]{\textsc{IncrmentalScore}}
        }
	    \lElseIf{$\widehat{h}(\theta_i) \le h(\theta^*)$}{
		    \Break
            \tcp*[f]{\textsc{EarlyTermination}}
	    }
    }
    \Return $\join(\theta^*)$\;

\caption{Preference-driven similarity join framework.}
\label{algo:general-1}
\end{algorithm2e}

\subsection{Pivotal Thresholds}\label{threshold}
Not every threshold has a chance to lead to the maximum preference score. In this section, we study how to identify a small set of thresholds $\Theta$ such that the maximum preference score can be obtained by only evaluating $\Theta$, i.e., $\max_{\theta \in [0, 1]} h(\theta) = \max_{\theta \in \Theta} h(\theta)$.

Consider $\Theta = \condset{sim(r, s)}{ r \in R, s \in S, r \in \topk^R(s)~\land~s \in \topk^S(r)}$.
Clearly, $\Theta$ is often dramatically smaller than $\abs{R} \times \abs{S}$, as
\begin{align*}
\abs{\Theta}  \le & \min \Big\{\abs{\condset{sim(r, s)}{ r \in R, s \in S \text{ where } s \in \topk^S(r)}}, \\
&  \abs{\condset{sim(r, s)}{r \in R, s \in S \text{ where } r \in \topk^R(s) }} \Big\} \le  \min\set{\abs{R}, \abs{S}}
\end{align*}
For example, in Figure~\ref{fig:toy/table}, $\Theta = \set{1, 0.8, 0.667}$.

We can show that both \pone and \ptwo have the same set of pivotal thresholds.
The basic idea is that, for any $\theta \notin \Theta$, there exists $\theta' > \theta$ such that $h(\theta') \ge h(\theta)$.
Remind that, through the paper, we only need to discuss the thresholds within the finite set $\condset{sim(r,s)}{r \in R, s \in S}$ as discussed in Section~\ref{sec:problem}.
\begin{lemma}
\label{lemma:cluster/threshold}
Given a threshold $\theta$, if $r \not\in \topk^R(s)$ or $s \not\in \topk^S(r)$ for any $(r, s) \in \join^=(\theta)$, then $\exists \theta' > \theta: h_c(\theta') \ge h_c(\theta)$.
\end{lemma}
\begin{proof}
Let $\theta'$ be a threshold such that $join(\theta) \setminus join(\theta') = join^=(\theta)$.
Bipartite $G_{\theta}$ can be derived by adding those new edges in $join^=(\theta)$ to bipartite $G_{\theta'}$.
For any $(r, s) \in \join^=(\theta)$, if $r \not\in \topk^R(s)$, then $r$ must be already in a non-trivial connected component of $G_{\theta'}$, and $s$ can only be added to the non-trivial connected component where $r$ belongs to. Similar situation happens if $s \not\in \topk^S(r)$. Since there is not any new non-trivial connected component in $G_{\theta}$ comparing to $G_{\theta'}$, $h_c(\theta') \ge h_c(\theta)$.
\end{proof}
\begin{lemma}
\label{lemma:outerJoin/threshold}
Given a threshold $\theta$, if $r \not\in \topk^R(s)$ and $s \not\in \topk^S(r)$ for any $(r, s) \in \join^=(\theta)$, then $\exists \theta' > \theta: h_o(\theta') \ge h_o(\theta)$.
\end{lemma}
\begin{proof}
Let $\theta'$ be a threshold such that $join(\theta) \setminus join(\theta') = join^=(\theta)$.
When $r \not\in \topk^R(s)~\lor~s \not\in \topk^S(r)$ for any $(r, s) \in \join^=(\theta)$,
\begin{align*}
& \abs{join(\theta)} - \abs{join(\theta')} = \abs{join^=(\theta)} \\
\ge& \abs{\condset{(r, s) \in join^=(\theta)}{s \in \topk^S(r)}} + \abs{\condset{(r, s) \in join^=(\theta)}{r \in \topk^R(s) }} \\
\ge& \abs{\condset{r \in cover^R(\theta)}{s \in S \text{ where } sim(r, s) = \theta~\land~s \in \topk^S(r) }} \\
+& \abs{\condset{s \in cover^S(\theta)}{r \in R \text{ where } sim(r, s) = \theta~\land~r \in \topk^S(s) }} \\
=& \abs{cover^R(\theta)} - \abs{cover^R(\theta')} + \abs{cover^S(\theta)} - \abs{cover^S(\theta')}
\end{align*}

Thus,
$h_o(\theta') - h_o(\theta) \ge 0$.
\end{proof}

Thanks to the existing fast top-$k$ similarity search algorithms~\cite{DBLP:conf/icde/XiaoWLS09,DBLP:journals/kais/XuGPWA16}, obtaining the pivotal thresholds can be efficient by computing the most similar objects of each object in $R$ and $S$, respectively. According to the above lemmas, it is guaranteed that the largest threshold having the maximum score in the finite set of thresholds is always in $\Theta$.

\subsection{Incremental Similarity Join}
\label{sec:algorithm/incsimjoin}
In this section, we present an efficient algorithm that incrementally computes the similarity join result $join(\theta_i)$ w.r.t.\ threshold $\theta_{i}$.
We do not need to conduct a similarity join for each pivotal threshold.
Since $join(\theta) = \cup_{\theta' \ge \theta} join^=(\theta')$, for each threshold $\theta$, we only need to compute the respective newly joined pairs $join^=(\theta)$. Thus, we can enumerate the thresholds in the value descending order, and compute the respective join results incrementally.

The algorithm consists of two steps. First, the algorithm incrementally computes a set of candidate pairs $cand(\theta_i)$. Second, the algorithm evaluates the similarity of each candidate pair and returns the pairs whose similarity values are at least the threshold $\theta_i$. While this two-step approach has been used by existing similarity join algorithms~\cite{DBLP:conf/icde/ChaudhuriGK06,DBLP:conf/www/BayardoMS07,DBLP:conf/www/XiaoWLY08,DBLP:conf/sigmod/WangLF12}, our contribution is a novel optimization technique, called lazy evaluation, which lazily evaluate the similarity of each candidate pair and reduces the cost. 

We focus on set-based similarity functions in this paper. Similar strategies can be applied to other kinds of similarity functions, like string-based or vector-based. Note that multi-set (bag) can also be used here instead of set.
Table~\ref{tab:similarity} shows the definitions of the similarity functions. Jaccard, overlap, dice, and cosine similarity are widely adopted in existing similarity join literature~\cite{DBLP:conf/icde/ChaudhuriGK06,DBLP:conf/www/BayardoMS07,DBLP:conf/www/XiaoWLY08,DBLP:conf/sigmod/WangLF12}. 
In addition, we include Tversky similarity~\cite{Scholar:tversky1977similarity}, which is a special asymmetric set-based similarity with different weights $\alpha$ and $1 - \alpha$ on $r$ and $s$, respectively. This similarity function is very useful in certain scenarios, like matching a text with an entity contained by the text.
\begin{table}[t]
\small
\centering
\caption{Summary of set-based similarity functions. 
Here $bound^{min}_{i,j} = \abs{r[:i] \cap s[:j]} \le \abs{r \cap s}$ and $bound_{i,j}^{max} = \abs{r[:i] \cap s[:j]} + \min\set{\abs{r} - i, \abs{s} - j} \ge \abs{r \cap s}$.
}
\label{tab:similarity}

\setlength{\tabcolsep}{0.1\tabcolsep}
\begin{tabular}{l|ccc|c}
Similarity & $sim$ & $sim_{i, j}^{min}$ & $sim_{i, j}^{max}$ & $t_{\theta}$ \\
\hline
Jaccard & $\frac{\abs{r \cap s}}{\abs{r} + \abs{s} - \abs{r \cap s}}$ & $\frac{bound^{min}_{i,j}}{\abs{r} + \abs{s} - bound^{min}_{i,j}}$ & $\frac{bound_{i,j}^{max}}{\abs{r} + \abs{s} - bound^{max}_{i,j}}$ & $\ceil{\theta \cdot \abs{r}}$ \\
Overlap & $\frac{\abs{r \cap s}}{\max\set{\abs{r}, \abs{s}}}$ & $\frac{bound^{min}_{i,j}}{\max\set{\abs{r}, \abs{s}}}$ & $\frac{bound_{i,j}^{max}}{\max\set{\abs{r}, \abs{s}}}$ & $\ceil{\theta \cdot \abs{r}}$ \\
Dice &  $\frac{2 \cdot \abs{r \cap s}}{\abs{r} + \abs{s}}$ & $\frac{2 \cdot bound^{min}_{i,j}}{\abs{r} + \abs{s}}$ & $\frac{2 \cdot bound_{i,j}^{max}}{\abs{r} + \abs{s}}$ & $\ceil{\frac{\theta}{2} \cdot \abs{r}}$ \\
Cosine &  $\frac{\abs{r \cap s}}{\sqrt{\abs{r} \cdot \abs{s}}}$ & $\frac{bound^{min}_{i,j}}{\sqrt{\abs{r} \cdot \abs{s}}}$ & $\frac{bound_{i,j}^{max}}{\sqrt{\abs{r} \cdot \abs{s}}}$ & $\ceil{\theta^2 \cdot \abs{r}}$ \\
Tversky &  $\frac{\abs{r \cap s}}{\alpha \cdot \abs{r} + (1 - \alpha) \cdot \abs{s}}$ & $\frac{bound^{min}_{i,j}}{\alpha \cdot \abs{r} + (1 - \alpha) \cdot \abs{s}}$ & $\frac{bound_{i,j}^{max}}{\alpha \cdot \abs{r} + (1 - \alpha) \cdot \abs{s}}$ & $\ceil{\theta \cdot \alpha \cdot \abs{r}}$ \\
\end{tabular}
\setlength{\tabcolsep}{10\tabcolsep}
\end{table}
Given an object $r$ as a set, we use $r[:i]$ to denote the first $i$ elements of $r$ assuming a global ordering of elements in the set.

\subsubsection{Candidate Pair Generation}
Established by prefix filtering~\cite{DBLP:conf/icde/ChaudhuriGK06}, if $sim(r, s) \ge \theta$, the number of elements in the overlap of the sets $\abs{r \cap s}$ is no fewer than an overlap threshold $t_\theta$ w.r.t. $\abs{r}$, where the overlap thresholds for set-based similarities are shown in Table~\ref{tab:similarity}. Thus, the candidate pair generation problem is converted to how to filter out the pairs with fewer than $t_\theta$ common elements.

To filter out the pairs with less than $t_\theta$ common elements, we fix a global ordering on the elements of all the objects, and sort the elements in each object based on the ordering. Like~\cite{DBLP:conf/sigmod/WangLF12}, we use the inverse document frequency as the global ordering. Prefix filtering~\cite{DBLP:conf/icde/ChaudhuriGK06} establishes that if $\abs{r \cap s} \ge t_{\theta}$, then $r[:\#prefix_{\theta}(r)] \cap s[:\#prefix_{\theta}(r)] \ne \emptyset$, where $\#prefix_{\theta}(r) = \abs{r} - t_{\theta} + 1$. 

Using an inverted index, we do not need to enumerate each pair $(r, s)$ to verify whether $r[:\#prefix_{\theta}(r)] \cap s[:\#prefix_{\theta}(r)] \ne \emptyset$. An inverted index maps an element to a list of objects containing the element. After building the inverted index for $S$, for each $r \in R$, we only need to merge the inverted lists of the elements in $r[:\#prefix_{\theta}(r)]$ to retrieve each $s \in S$ such that $r[:\#prefix_{\theta}(r)] \cap s[:\#prefix_{\theta}(r)] \ne \emptyset$.

Our goal is to generate the candidate pairs for $[\theta_i, \theta_{i-1})$.
We use an incremental prefix filtering approach~\cite{DBLP:conf/icde/XiaoWLS09} that memorizes previous results to avoid regenerating the candidate pairs for $[\theta_{i-1}, 1]$.

\subsubsection{Lazy Evaluation}
Suppose we want to check whether the similarity of a candidate pair $(r, s) \in cand(\theta_i)$ is no smaller than a threshold $\theta_{i}$ or not. The idea of lazy evaluation is to iteratively compute both a maximum and a minimum possible value of $sim(r, s)$, denoted by $sim^{max}(r, s)$ and $sim^{min}(r,s)$, respectively. Interestingly, both $sim^{max}(r, s)$ and $sim^{min}(r,s)$ get tighter through the process, and finally converge at $sim(r, s)$. During this process, we use the values for lazy evaluation in two ways.
\begin{itemize}
\item
If $sim^{max}(r, s) < \theta_i$, then $sim(r, s) < \theta_i$. Thus, it is only necessary to resume evaluating $(r, s)$ for a future smaller threshold $\theta_{j}$ where $\theta_{j-1} > sim^{max}(r, s) \ge \theta_j$.
\item
If $sim^{min}(r, s) \ge \theta_i$, then $\theta_{i - 1} > sim(r, s) \ge \theta_i$. Thus, $(r, s)$ does not need to be fully evaluated at all.
\end{itemize}

We scan $r$ and $s$ iteratively together from left to right, according to the global ordering. Assuming $r[:i]$ and $s[:j]$ have been scanned, Table~\ref{tab:similarity} shows the maximum/minimum possible values of $sim(r, s)$.
Through the scanning,
we iteratively update $sim^{max}(r, s)$ and $sim^{min}(r, s)$ accordingly.

The pseudo-code of the lazy evaluation-powered algorithm is shown in Algorithm~\ref{algo:general}. The algorithm first computes $cand(\theta_i)$, by generating a set of candidate pairs for $[\theta_i, \theta_{i-1})$ together with the previously postponed candidate pairs from larger thresholds. Then, the algorithm examines each candidate pair in $cand(\theta_i)$ and decides whether it should be added into $\join^=(\theta_i)$ or postponed and added into $cand(\theta_j)$ for a smaller threshold $\theta_j$ (found by binary search over the rest of the thresholds). Finally, $\join^=(\theta_i)$ is returned.

\begin{algorithm2e}[t]
\small
\DontPrintSemicolon

\SetKwInput{Prototype}{Function}
\SetKwInput{Input}{Input}
\SetKwInput{Output}{Output}

\SetKw{LogicalAnd}{and}
\SetKw{LogicalOr}{or}
\SetKw{LogicalNot}{not}
\SetKw{Continue}{continue}
\SetKw{Break}{break}
\SetKw{True}{true}
\SetKw{False}{false}
\Input{thresholds $\theta_{i-1}$ and $\theta_{i}$ where $\theta_{i-1} > \theta_i$}
\Output{incremental similarity join result $\join^=(\theta_i)$}
\BlankLine
	$\join^=(\theta_i) \leftarrow \emptyset$ \;
	let $cand(\theta_i)$ be the candidate pairs for $[\theta_i, \theta_{i-1})$\;

    \ForEach{pair $(r, s) \in  cand(\theta_i)$}{
        \While{$sim^{min}(r, s) < \theta_i \le sim^{max}(r, s)$}{
            update $sim^{max}(r, s)$ and $sim^{min}(r, s)$\;
        }
        \If{$sim^{max}(r, s) < \theta_i$}{
            find $\theta_j: \theta_{j - 1} > sim^{max}(r, s) \ge \theta_j$ by binary search\;
            add $(r, s)$ into $cand(\theta_j)$\;
        }
        \Else{
            add $(r, s)$ into $\join^=(\theta_i)$\;
        }
    }

    \Return $\join^=(\theta_i)$\;

\caption{Incremental similarity join.}
\label{algo:general}
\end{algorithm2e}

\subsection{Incremental Score Computation}
\label{sec:algorithm/incscore}

For both our preferences, computing the preference score for each similarity threshold is straightforward. However, as we need to compute the preference scores for multiple thresholds, it is necessary to explore how to further reduce the cost.

For preference \pone,
when computing the join result incrementally by decreasing $\theta$,
if two objects of each newly joined pair in $join^=(\theta)$ are in different connected components, the connected components are merged together to form a larger connected component.
We use a disjoint-set data structure to dynamically track newly joined pairs, and update the connected components accordingly. It only takes almost $O(1)$ amortized time~\cite{Scholar:dynamicgraph} for each newly joined pair in $join^=(\theta)$.

For preference \ptwo, when processing incrementally by the value decreasing order of $\theta$, we only need to scan each $join^=(\theta)$ to update $join(\theta)$, $cover^R(\theta)$, $cover^S(\theta)$, and further the preference score. It only takes $O(1)$ time for each newly joined pair in $join^=(\theta)$.

\subsection{Early Termination}\label{termination}
The goal of early termination is to determine if we can return the current most preferred result without evaluating the remaining thresholds. In our algorithm, the thresholds are evaluated in the descending order. Suppose threshold $\theta_i$ has just been evaluated. At this point, we have known the preference score for each threshold that is at least $\theta_i$. Let $h(\theta^*)$ denote the current best preference score, i.e.,
$h(\theta^*) = \max_{\theta \geq \theta_i} h(\theta)$.
If we can derive an upper-bound $\widehat{h}(\theta_i)$ of the preference scores for the remaining thresholds and show that the upper-bound is no larger than $h(\theta^*)$, then it is safe to stop at $\theta_i$ and $h(\theta^*)$ is the best result overall.

For \pone, as the threshold decreases, a previously unseen non-trivial connected component can only be created by merging two trivial connected components. Since a new non-trivial connected component contains at least one object from $R \setminus \cover^R(\theta_i)$ and one object from $S \setminus \cover^S(\theta_i)$, the number of non-trivial connected components in any $G_{\theta'}$ such that $\theta' < \theta_i$ is at most $\min\set{\abs{R \setminus \cover^R(\theta_i)}, \abs{S \setminus \cover^S(\theta_i)}}$. Thus, an upper-bound is
\[\widehat{h_c}(\theta_i) = h_c(\theta_i)+ \min\set{\abs{R \setminus \cover^R(\theta_i)}, \abs{S \setminus \cover^S(\theta_i)}}\]

For \ptwo, as the threshold decreases, the join result $\join(\theta_i)$ includes more pairs. Whenever a new pair $(r, s)$ is joined, $\abs{\join(\theta_i)}$ increases by one. The only way to get the preference score increased by $1$ is that $\abs{\cover^R(\theta_i)}$ and $\abs{\cover^S(\theta_i)}$ both increase by $1$. In this case, $r$ has to be come from $R \setminus \cover^R(\theta_i)$ and $s$ has to be come from $S \setminus \cover^S(\theta_i)$. Therefore, the preference score can at most increase by $\min\set{\abs{R \setminus \cover^R(\theta_i)}, \abs{S \setminus \cover^S(\theta_i)}}$. Accordingly, we set an upper-bound to
\[\widehat{h_o}(\theta_i) = h_o(\theta_i)+ \min\set{\abs{R \setminus \cover^R(\theta_i)}, \abs{S \setminus \cover^S(\theta_i)}}\]

\section{Experimental Results}
\label{sec:exp}

We present a series of experimental results in this section. The programs were implemented in Python running with PyPy\footnote{PyPy~(\url{http://pypy.org/}) is an advanced just-in-time compiler, providing $10$ times faster performance for our algorithm than the standard Python interpreter.}. The experiments were conducted using a Mac Pro Late $2013$ Server with Intel Xeon $3.70$GHz CPU, $64$GB memory, and $256$GB SSD.

\subsection{Datasets}
We adopt four real-world web datasets with ground-truth for evaluation. Table~\ref{tab:exp/dataset} shows the characteristics of each dataset.
\begin{table}[t]
\small
\centering
\caption{Dataset characteristics.}
\label{tab:exp/dataset}
\setlength{\tabcolsep}{0.1\tabcolsep}
\begin{tabular}{l|ccc:ccc:c}
Dataset & \multicolumn{3}{c:}{$R$} & \multicolumn{3}{c:}{$S$} & $\abs{\mathbb{C^+}}$ \\
& $\abs{R}$ & Max. Len. & Avg. Len. & $\abs{S}$ & Max. Len. & Avg. Len. & \\
\hline
\textsf{Wiki Editors} & $2{,}239$ & $20$ & $9.65$ & $1{,}922$ & $16$ & $8.66$ & $2{,}455$ \\
\textsf{Restaurants} & $533$ & $96$ & $48.38$ & $331$ & $91$ & $43.5$ & $112$ \\
\textsf{Scholar-DBLP} & $64{,}259$ & $259$ & $115.9$ & $2{,}562$ & $326$ & $106.61$ & $5{,}347$ \\
\textsf{Wiki Links} & $187{,}122$ & $1{,}393$ & $17.08$ & $168{,}652$ & $209$ & $16.35$ & $202{,}272$ \\
\end{tabular}
\setlength{\tabcolsep}{10\tabcolsep}

\end{table}

\textsf{Wiki Editors}~\cite{URL:roger} is about misspellings of English words, made by Wikipedia page editors where the errors are mostly typos. Each misspelling has at least one correct word. $R$ contains the misspellings, while $S$ contains the correct words. The ground-truth $\mathbb{C^+}$ are pairs of each misspelling and the corresponding correct word.

\textsf{Restaurants}~\cite{URL:riddle} links the restaurant profiles between two websites. Each profile contains the name and address of a restaurant. We remove the phone number and cuisine type, which are available in the original data, to make it more challenging. $R$ and $S$ are profiles of restaurants, and the ground-truth $\mathbb{C^+}$ identifies the pairs of profiles linking the same restaurants. 
Every restaurant has at most one match.

\textsf{Scholar-DBLP}~\cite{URL:matching} finds the same publications in Google Scholar and DBLP, where each record in DBLP has at least one matching record in Google Scholar. Each record on both websites contains the title, author names, venue, and year. $R$ and $S$ are publications identified by Google Scholar and DBLP, respectively, and the ground-truth $\mathbb{C^+}$ are pairs of records linking the same publications. 

\textsf{Wiki Links}~\cite{URL:wiki_links} is a large dataset containing short anchor text on web pages and the Wikipedia link that each anchor text contains. $R$ contains the anchor text, while $S$ contains the Wikipedia entities extracted from Wikipedia links. For example, link \url{https://en.wikipedia.org/wiki/Harry_Potter_(character)} is converted to entity ``Harry Potter''. The ground-truth $\mathbb{C^+}$ are pairs linking each anchor text and its entity. Each anchor text may also contain the text of the entity that the Wikipedia link refers to, such as ``\dots\underline{Harry Potter} is a title character\dots''. This dataset has multiple files, where each contains roughly $200,000$ mappings (about 200 MB). We used the first file in most of our evaluation except for the  scalability evaluation where the first four files were used.

To adopt set-based similarities, each string needs to be converted to a set. Empirically, people often convert a long string into a bag of words and convert a short string into a set of grams. For the first dataset, since a string represents a word, we convert each string into a bag of $2$-grams. For the rest three datasets, since the strings are much longer, we convert each string into a bag of words. We use Jaccard similarity for the first three datasets. For the \textsf{Wiki Links} dataset, since anchor text often contains both the entity text and some other irrelevant texts, we choose Tversky similarity ($\alpha = 0.1$).

\subsection{Threshold-driven vs. Preference-driven}

In this section, we empirically investigate the pros and cons of both our preference-driven approach and other possible solutions for tuning a similarity join threshold.

We demonstrate the sensitivity of thresholds w.r.t.\ $F_1$-score in Table~\ref{tab:exp/accuracy/classic}. A small deviation from the optimal threshold affects the result quality dramatically. This clearly shows that threshold tuning is crucial for threshold-driven similarity join. On two datasets, our preference-driven approach can achieve the optimal $F_1$-score, and on the other two datasets, they are very close to the optimal.   
\begin{table}[t]
\small
\centering
\caption{Sensitivity of thresholds on $F_1$-score. Three thresholds are used: the optimal one $\theta^{\#}$, the higher one $\theta^{+} = \min\set{1, \theta^{\#} + 0.1}$, and the lower one $\theta^{-} = \max\set{0, \theta^{\#} - 0.1}$.
Numbers in red are when ours achieves optimal $F_1$-scores.
}
\label{tab:exp/accuracy/classic}
\setlength{\tabcolsep}{0.5\tabcolsep}
\begin{tabular}{l|c:c:c|cl}
Dataset & using $\theta^{+}$ & using $\theta^{\#}$ & using $\theta^{-}$ & \multicolumn{2}{c}{Preference-driven} \\
\hline
\textsf{Wiki Editors} & $0.599$ & $0.764$ & $0.705$ & ${\highlight{0.764}}$ & using $h_c$ and $h_o$\\
\textsf{Restaurants} & $0.725$ & $0.816$ & $0.597$ & ${0.809}$ & using $h_o$\\
\textsf{Scholar-DBLP} & $0.697$ & $0.759$ & $0.554$ & ${0.754}$ & using $h_c$\\
\textsf{Wiki Links} & $0.774$ & $0.780$ & $0.624$ & ${\highlight{0.780}}$ & using $h_c$\\
\end{tabular}
\setlength{\tabcolsep}{2\tabcolsep}
\end{table}

\begin{figure*}[t]
\centering{
\subfloat{
\includegraphics[scale=0.75]{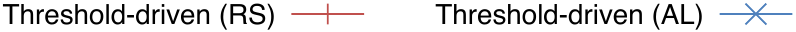}
\quad
\includegraphics[scale=0.75]{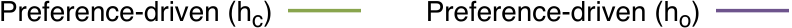}
}
\setcounter{subfigure}{0}
\\
\captionsetup[subfigure]{labelformat=empty}
\subfloat[{\textsf{Wiki Editors}}]{
\includegraphics[scale=0.75]{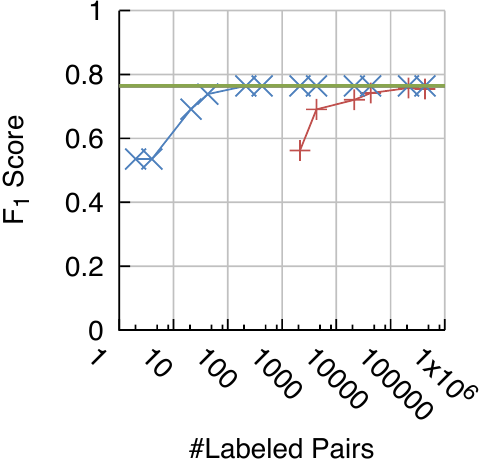}
}
~
\subfloat[{\textsf{Restaurants}}]{
\includegraphics[scale=0.75]{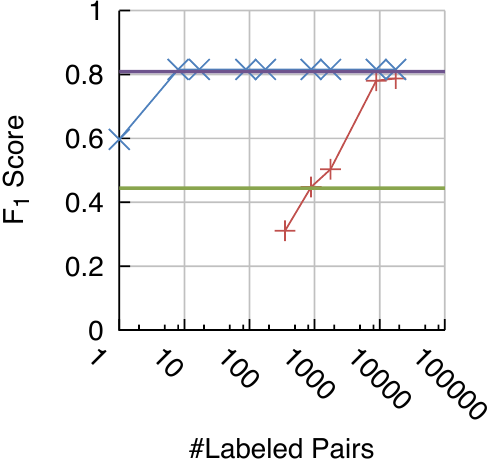}
}
~
\subfloat[\textsf{Scholar-DBLP}]{
\includegraphics[scale=0.75]{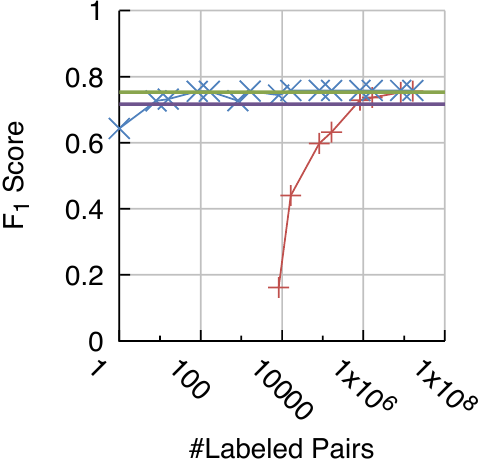}
}
~
\subfloat[{\textsf{Wiki Links}}]{
\includegraphics[scale=0.75]{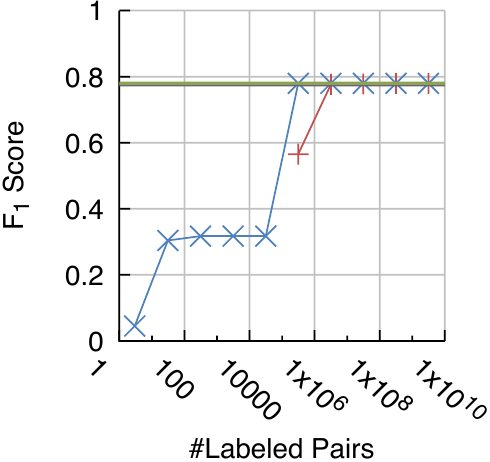}
}
}
\caption{Comparisons of accuracy between threshold-driven (using random sampling (RS) or active learning (AL) for threshold tuning) and preference-driven approaches.}\label{fig:exp/accuracy/sample}
\end{figure*}

\begin{figure}[t]
\centering{
\subfloat{
\includegraphics[scale=0.75]{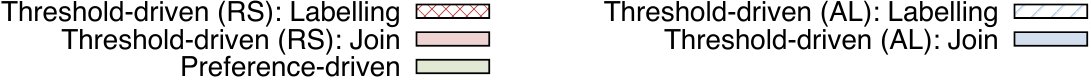}
}
\setcounter{subfigure}{0}
\\
\subfloat[Running time.]{
\includegraphics[scale=0.75]{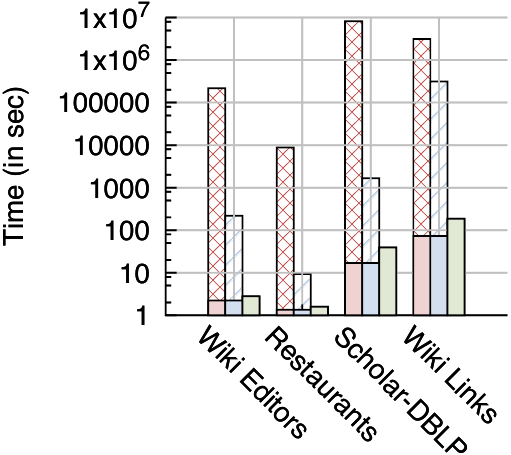}
}
~
\subfloat[Peak memory.]{
\includegraphics[scale=0.75]{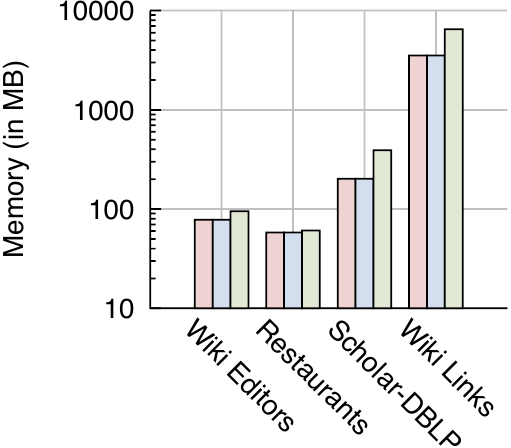}
}
}
\caption{Comparisons of efficiency between threshold-driven and preference-driven approaches. For threshold-driven approach, we tune the threshold such that it achieves the closest (within $0.01$) $F_1$-score as preference-driven approach.
We assume that each pair needs $1s$ to label.
The same preference in Table~\ref{tab:exp/accuracy/classic} is used here for each dataset.}\label{fig:exp/efficiency/human}
\end{figure}

To compare with human labeling approaches, we adopt two supervised approaches to tune a similarity threshold. The first one is random sampling, where pairs are randomly sampled and labeled. We use the same sampling method as \cite{DBLP:conf/sigmod/GokhaleDDNRSZ14}, where the larger set $R$ is sampled with a specified sampling rate, and then joined with $S$ to derive all the pairs to be labeled. Apparently, at least $\abs{S}$ pairs need to be labeled. The second approach~\cite{DBLP:conf/cikm/BuchA15} uses active learning to tune a threshold by incrementally labeling the most uncertain pair to the classifier. When there is a tie, the one with the highest similarity is selected. The threshold that achieves the best on the labeled data is selected. Figure~\ref{fig:exp/accuracy/sample} shows the results.   
For random sampling, only using a very high sampling rate like $10\%$ can almost catch up with our method. For active learning, the number of necessary labels is significantly reduced, however, still hundreds to thousands of pairs need to be labeled on most of the datasets.

The total running time of the two supervised approaches contains two parts: the labeling part which tunes the threshold, and the joining part using the tuned threshold. In comparison, our approach returns both a threshold and its join results in a unified framework. We calculate the end-to-end time that the threshold-driven approach needs in order to achieve the same quality as our preference-driven approach. For simplicity, we assume that each pair takes $1s$ to be labeled correctly by a human, which is a very conservative estimation. 

Figure~\ref{fig:exp/efficiency/human} shows the results. 
The labeling step is much more costly than the joining step. The end-to-end time of the preference-driven approach is orders of magnitude faster. 
The preference-driven approach uses one to two times more memory than a threshold-driven approach, due to the need of caching the candidate pairs.

\subsection{Accuracy of Preference-driven Approach}

In this section, we evaluate the accuracy of our preference-driven approach.
Table~\ref{tab:exp/accuracy} compares the results between $h_c$ and $h_o$. The precision, recall, and $F_1$-score are presented, together with the preferred threshold $\theta^*$. For \textsf{Wiki Editors}, due to its nature of high similarity between misspellings and correct words, both preferences return the same optimal result.
For the record-linkage task on \textsf{Restaurants}, 
$h_o$ gives the significantly better result as it favors one-to-one matching.
For \textsf{Scholar-DBLP}, $h_c$ gives the better result, because maximizing the number of non-trivial connected components actually satisfies the nature of many-to-many matching. For the \textsf{Wiki Links}, $h_c$ gives the optimal result, and $h_o$ is quite close.

\begin{table}[t]
\small
\centering
\caption{Accuracy for varying datasets and preferences.}
\label{tab:exp/accuracy}
\setlength{\tabcolsep}{0.2\tabcolsep}
\begin{tabular}{l|cccc:cccc}
Dataset & \multicolumn{4}{c:}{$h_c$ (\pone)} & \multicolumn{4}{c}{$h_o$ (\ptwo)} \\
\hline
\textsf{Wiki Editors} & $0.625$ & $0.837$ & $0.704$ & $\mathbf{{0.764}}$ & $0.625$ & $0.837$ & $0.704$ & $\mathbf{{0.764}}$ \\
\textsf{Restaurants} & $0.429$ & $0.291$ & $0.938$ & $0.444$ & $0.556$ & $0.805$ & $0.813$ & $\mathbf{0.809}$ \\
\textsf{Scholar-DBLP} & $0.361$ & $0.841$ & $0.683$ & $\mathbf{0.754}$ & $0.419$ & $0.903$ & $0.595$ & $0.717$ \\
\textsf{Wiki Links} & $0.957$ & $0.959$ & $0.657$ & $\mathbf{{0.780}}$ & $0.972$ & $0.968$ & $0.648$ & $0.776$ \\
\hline
& $\theta^*$ & Precision & Recall & $F_1$ & $\theta^*$ & Precision & Recall & $F_1$
\end{tabular}
\setlength{\tabcolsep}{5\tabcolsep}
\end{table}

\subsection{Efficiency of Preference-driven Approach}

There is no existing work solving the same problem. We choose the brute-force method in Section~\ref{sec:framework} as a baseline to evaluate the efficiency. This method takes almost the same time and memory regardless of the preference due to the same amortized time for processing each newly joined pair incrementally. 

Figure~\ref{fig:exp/efficiency}(a) shows the number of thresholds evaluated by the baseline and our algorithm. On most of the datasets, our algorithm evaluates $10$ to $100$ times less thresholds than the baseline. 
Our method achieves a significant speedup due to the combination of other optimization techniques (i.e., incremental similarity join and early termination).  As Figure~\ref{fig:exp/efficiency}(b) shows, our algorithm is $10$ to $100$ times faster than the baseline on all the large datasets.
Figure~\ref{fig:exp/efficiency}(c) shows that our method consumes significantly less memory.

\begin{figure*}[t]
\begin{minipage}{.74\textwidth}
\centering{
\subfloat{
\includegraphics[scale=0.75]{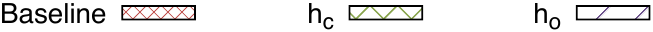}
}
\setcounter{subfigure}{0}
\\
\subfloat[Number of evaluated thresholds.]{
\includegraphics[scale=0.75]{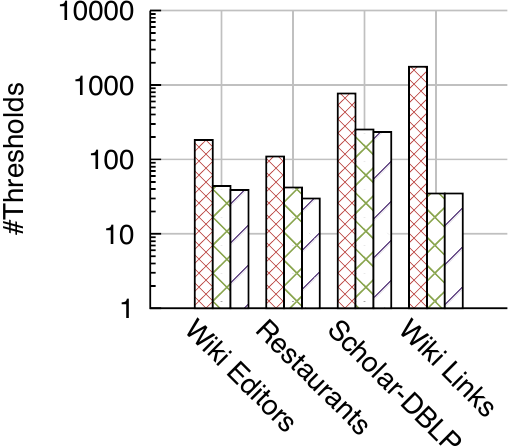}
}
~
\subfloat[Running time.]{
\includegraphics[scale=0.75]{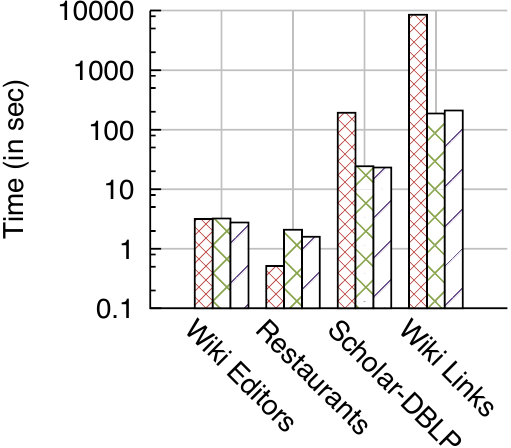}
}
~
\subfloat[Peak memory.]{
\includegraphics[scale=0.75]{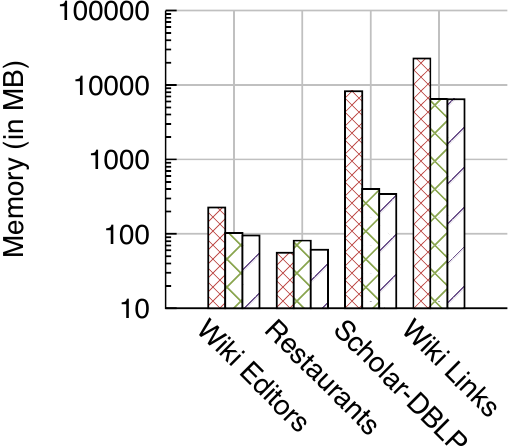}
}
}
\caption{Efficiency of preference-driven approach. }\label{fig:exp/efficiency}
\end{minipage}
\begin{minipage}{.24\textwidth}
\vspace{-.1in}
\subfloat{
\hspace{-.5in}
\includegraphics[scale=0.75]{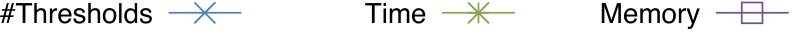}
}
\setcounter{subfigure}{0}
\\
\subfloat{
\includegraphics[scale=0.75]{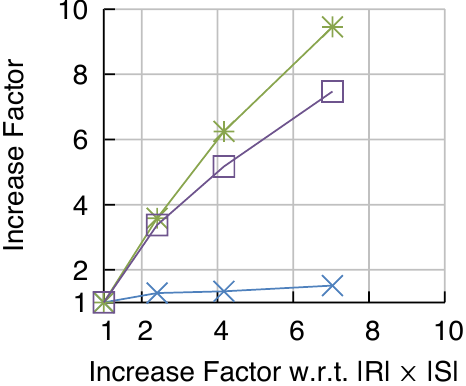}
}
\\
\caption{Scalability on \textsf{Wiki Links} using $h_c$.}\label{fig:exp/scale}
\end{minipage}
\end{figure*}

For incremental similarity join, we use lazy evaluation for speedup.
Instead, a simple approach computes the exact similarities for all the new candidate pairs produced by each threshold, and puts them into a max-heap. Those pairs in the heap whose similarities are no less than the current threshold are popped out for evaluation.
On smaller datasets \textsf{Wiki Editors} and \textsf{Restaurants}, the simple approach works slightly better, as there are not many pairs that need lazy evaluation. However, on larger datasets \textsf{Scholar-DBLP} and \textsf{Wiki Links}, our lazy evaluation approach is $2$ to $5$ times faster.

For scalability, we evaluate our algorithm using $h_c$ on a larger version of \textsf{Wiki Links} that contains 4 files.
There are overlapping objects between different parts.
Figure~\ref{fig:exp/scale}
shows the scalability results on number of evaluated thresholds, running time, and peak memory usage. Our method is scalable.

\section{Related Work}
\label{sec:related}

Due to the crucial role of similarity join in data integration and data cleaning, a large number of similarity join algorithms have been proposed~\cite{DBLP:conf/vldb/ArasuGK06,DBLP:conf/icde/ChaudhuriGK06,DBLP:conf/www/BayardoMS07,DBLP:conf/www/XiaoWLY08,DBLP:journals/pvldb/XiaoWL08,DBLP:conf/icde/XiaoWLS09,DBLP:journals/pvldb/WangLF10,DBLP:journals/pvldb/LiDWF11,DBLP:conf/sigmod/WangLF12,DBLP:journals/pvldb/JiangLFL14}. There are also scalable implementations of the algorithms using the MapReduce framework~\cite{DBLP:conf/sigmod/VernicaCL10,DBLP:journals/pvldb/MetwallyF12,DBLP:conf/icde/DengLHWF14}. Top-k similarity join is also explored~\cite{DBLP:conf/icde/XiaoWLS09,DBLP:journals/kais/XuGPWA16,DBLP:conf/icde/KimS12}.

While the majority of the existing work on similarity join needs to specify a similarity threshold or a limit of the number of results returned, there do exist some studies that seek to find a suitable threshold for similarity join in a supervised way~\cite{DBLP:journals/pvldb/WangLYF11,DBLP:conf/vldb/ChaudhuriCGK07,DBLP:conf/sigmod/ArasuGK10,DBLP:conf/cikm/BuchA15}. 
Both \cite{DBLP:conf/sigmod/ArasuGK10} and~\cite{DBLP:conf/cikm/BuchA15} adopted active learning to tweak the threshold.
Chaudhuri \textit{et al.}~\cite{DBLP:conf/vldb/ChaudhuriCGK07} learned an operator tree, where each node contains a similarity threshold and a similarity function on each of the splitting attributes.
Wang \textit{et al.}~\cite{DBLP:journals/pvldb/WangLYF11} modeled this problem as an optimization problem and applied hill climbing to optimize the threshold-selection process.
Nevertheless, all these methods need humans to label a number of pairs, which are selected based on either random sampling or active learning. In comparison, our method does not need any labeled data.

There are studies about preferences in the database research (See \cite{DBLP:journals/tods/StefanidisKP11} for a thorough survey.). However, they mainly focus on how the joined tuples are ranked and selected, instead of how the tables (and in our case sets of objects) are joined to generate the joined tuples. Alexe~\textit{et al.}~\cite{DBLP:journals/pvldb/AlexeRT14} discussed user-defined preference rules for integrating temporal data, which is orthogonal to our work. To the best of our knowledge, \cite{DBLP:conf/iri/GaoPWC17} is the only work discussing result set preference for joining relational tables.

\section{Conclusions}
\label{sec:conclusion}

Threshold selection, usually neglected in the similarity join literature, can actually be a bottleneck in an end-to-end similarity join process.  To mitigate the challenge, we propose preference-driven similarity join in this paper. We formalize preference-driven similarity join, and propose the notion of result set preference. We present two specific preferences, \pone and \ptwo as proofs of concept. We develop a general framework for preference-driven similarity join and the efficient computation methods. 
We evaluate our approach on four real-world datasets from a diverse range of application scenarios. The results demonstrate that preference-driven similarity join can achieves high-quality results without any labeled data, and our proposed framework is efficient and scalable.   
Future directions include developing more preferences for various application scenarios and supporting various similarity functions.

\end{sloppy}
\end{document}